\documentclass[10pt]{article}
\usepackage{geometry}
\usepackage{fullpage}
\usepackage{times}
\usepackage[ruled]{algorithm2e}
\usepackage{algorithmic}
\usepackage[cmex10]{amsmath}
\usepackage{amsmath}
\usepackage{amssymb}
\usepackage{amsthm}
\usepackage{graphicx} 
\usepackage{subfig}
\usepackage{array}

\makeatletter
\newcommand{\rn}[1]{\expandafter\@slowromancap\romannumeral #1@}
\makeatother

\usepackage[backref=page]{hyperref}
\hypersetup{
	bookmarks=true,
	bookmarksnumbered=true,
	bookmarksopen=true,
	pdftitle={Optimal bandwidth-aware VM allocation for Infrastructure-as-a-Service},
	pdfauthor={Debojyoti Dutta, Michael Kapralov, Ian Post, Rajendra Shinde}
}

\DeclareMathOperator{\degree}{\textrm{degree}}
\renewcommand{\cong}{\textrm{Cong}}

\newtheorem{theorem}{Theorem}
\newtheorem{lemma}[theorem]{Lemma}
\newcommand{\etal}{{\em et al.}}

\begin{document}

\title{Optimal bandwidth-aware VM allocation for Infrastructure-as-a-Service}
\author{Debojyoti Dutta \\
Cisco Inc.\\
\href{mailto:dedutta@cisco.com}{dedutta@cisco.com}
\and
Michael Kapralov \\
Stanford University\\
\href{mailto:kapralov@stanford.edu}{kapralov@stanford.edu}
\and
Ian Post \\
Stanford University\\
\href{mailto:itp@stanford.edu}{itp@stanford.edu} 
\and
Rajendra Shinde \\
Stanford University\\
\href{mailto:rbs@stanford.edu}{rbs@stanford.edu}
}

\maketitle
\begin{abstract} 
Infrastructure-as-a-Service (IaaS) providers need to offer richer services to be competitive while optimizing their resource usage to keep costs down. Richer service offerings include new resource request models involving bandwidth guarantees between virtual machines (VMs). Thus we consider the following problem: given a VM request graph (where nodes are VMs and edges represent virtual network connectivity between the VMs) and a real data center topology, find an allocation of VMs to servers that satisfies the bandwidth guarantees for every virtual network edge---which maps to a path in the physical network---and minimizes congestion of the network.  

Previous work has shown that for arbitrary networks and requests, finding the optimal embedding satisfying bandwidth requests is $\mathcal{NP}$-hard. However, in most data center architectures, the routing protocols employed are based on a spanning tree of the physical network. In this paper, we prove that the problem remains $\mathcal{NP}$-hard even when the physical network topology is restricted to be a tree, and the request graph topology is also restricted. We also present a dynamic programming algorithm for computing the optimal embedding in a tree network which runs in time $O(3^kn)$, where $n$ is the number of nodes in the physical topology and $k$ is the size of the request graph, which is well suited for practical requests which have small $k$.  Such requests form a large class of web-service and enterprise workloads. Also, if we restrict the requests topology to a clique (all VMs connected to a virtual switch with uniform bandwidth requirements), we show that the dynamic programming algorithm can be modified to output the minimum congestion embedding in time $O(k^2n)$.
\end{abstract}

 \section{Introduction} \label{sec:intro}
Infrastructure-as-a-Service (IaaS) providers like Amazon~\cite{amazon}, Rackspace~\cite{rackspace} and Go-grid~\cite{gogrid} provide computing and other services on demand and charge based on usage. This has resulted in the commoditization of computing and storage. Typically, these providers provide service level agreements (SLA)~\cite{rackspace-sla} where they guarantee the type of virtual machines (VMs) that they provide and the amount of disk space available to these VMs. Although some providers offer additional services like dedicated firewalls and load-balancers, no network performance guarantees are provided, which are critical for workloads like content distribution networks, desktop virtualization, etc. 
Given the rapid growth and innovation in these services~\cite{aws-growth}, it is important for service providers (SPs) to offer innovative service models for differentiation, e.g., by offering richer network SLAs to be competitive while optimizing their resource usage to keep costs down.

Next generation cloud services will require improved quality of service (QoS) guarantees for application workloads. For example, multi-tier enterprise applications~\cite{tier-apps} require network isolation and QoS guarantees such as bandwidth guarantees, and for over-the-top content distribution using a cloud infrastructure, bandwidth, jitter and delay guarantees are important in determining performance. Similar guarantees are necessary for MapReduce-based analytics workloads too. Moreover, networking costs are currently a significant fraction of the total infrastructure cost in most data center (DC) designs~\cite{ibm-dc-network,cisco-dc-guide} since servers are cheap compared to core switches and routers. Thus, in order to provide richer network SLAs, it is important for SPs to ensure that networking resources are efficiently utilized 
while at the same time ensuring low congestion (that leads to better load balancing and more room for overprovisioning). 

In this paper we consider a virtualization request model in which clients can request bandwidth guarantees between pairs of virtual machines (VMs)~\cite{SecondNet} for which SPs will allocate resources within their infrastructure. 
This naturally leads us to study the following resource allocation problem: given a VM request graph---where nodes are VMs and edges represent virtual network connectivity between the VMs---and a real data center topology, find an allocation of VMs to servers that satisfies the bandwidth guarantees for every virtual network edge and minimizes congestion of the network. Note that in this setting, each virtual edge maps to a path in the physical network topology.

The above request graph model is driven by application workloads that execute on top of network infrastructure provided by the SPs. Common workloads include enterprise applications~\cite{tier-apps}, MapReduce~\cite{mapreduce}, and web hosting, and different workloads can lead to different service models. For instance, many web services request a small number of VMs to implement the web servers, the application servers, and the database. The VM implementing the web server receives a request and forwards it to an application server VM, which in turn queries the database server VMs. In such cases, specific bandwidth guarantees between the outside world and the web server, the web server and the application server, and so on, are important to ensure QoS.  In MapReduce workloads on the other hand, it has been shown that network optimization can yield better results than adding machines~\cite{mr-network}, but in this setting since all the VMs implementing map and reduce tasks communicate with each other via data shuffle, the aggregate bandwidth available to the VMs may determine the application performance.

A number of metrics have been studied to measure the network load including congestion, jitter, delay, hop count, or a combination of the above. Here we focus on minimizing congestion, but we also note that our algorithmic techniques are generic and can easily be adapted to optimize other metrics. 

It has been shown previously that the problem of embedding virtual requests in arbitrary networks is $\mathcal{NP}$-hard \cite{vne, SecondNet}. However in most data center networks, routing protocols used rely on a spanning tree of the physical network \cite{cisco-dc-guide}. Hence, in this paper we study the problem of minimizing network congestion while allocating virtual requests when the network topology is restricted to be a tree. 
 
\subsection{Our Contributions}
First, we prove that optimally allocating VMs remains $\mathcal{NP}$-hard even when both the physical network topology and request topology are highly restricted.  
We show that if the network topology is a tree then even for simple request topologies like weighted paths with the weights signifying the amount of bandwidth required between the corresponding VMs, it is $\mathcal{NP}$-hard to approximate the minimum congestion to a factor better than $O(\theta)$, where $\theta$ is the ratio of the largest to smallest bandwidth requirements in the path request. We also show that in the unweighted case (or uniform bandwidth requirement on all edges) the problem is $\mathcal{NP}$-hard to approximate to within a factor of $O(n^{1-\epsilon})$ for any $\epsilon\in (0, 1)$, even for the case when the request topology is a tree. 

Given these complexity results, we cannot hope for an efficient algorithm for all instances of the problem.  However, we note that in practice, many workloads consist of a small number of VMs allocated in a huge datacenter.
Accordingly, our second result is a dynamic programming algorithm (Algorithm \ref{alg:cong}) for computing the minimum congestion embedding of VMs in a tree network for any request graph, which satisfies the pairwise bandwidth requirements and runs in time $O(3^kn)$, where $n$ is the number of nodes in the physical topology and $k$ is the number of VMs in the request graph.  
Enterprise workloads often consist of small requests with specific bandwidth requirements between VMs, and for these instances the exponential $O(3^k)$ term is quite small, and can thus be optimally served using our algorithm whose run time is only linear in the network size.

Third, workloads like Map-Reduce jobs have too many VMs to use an algorithm with a runtime of $O(3^kn)$, but these have uniform bandwidth requirements between the VMs \cite{TrafficAware}, and we show that the exponential dependence on $k$ can be removed when the request network is uniform.
For the special case in where the requests are restricted to be cliques or {\em virtual clusters} \cite{sig11}, we propose an algorithm that finds the minimum congestion embedding in $O(k^2n)$ time (Algorithm \ref{alg:virtual-cluster}).
Hence our algorithms yield the minimum congestion embeddings of virtualization requests for several common use cases.

We also present simulations which validate our results for common request models and practical network configurations. 

\subsection{Outline of the paper}

The paper is organized as follows. We first review previous work in Section~\ref{subsec:related} and formally define the problem and notation in Section~\ref{sec:prelim}. We prove the hardness results in Section~\ref{sec:hardness} followed the algorithms in Section~\ref{sec:algo}. In Section~\ref{sec:simu} we provide simulations, which validate the running time and correctness of our algorithms. Finally, we conclude and point to future work in Section~\ref{sec:conclusion}.

\section{Related Work}
\label{subsec:related}


Previous work has shown that the problem of embedding virtual request graphs in arbitrary physical networks is $\mathcal{NP}$-hard \cite{vne,SecondNet}. A number of heuristic approaches have been proposed including mapping VMs to nodes in the network greedily and mapping the flows between VMs to paths in the network via shortest paths and multi-commodity flow algorithms \cite{fan, zhu}. However these approaches do not offer provable guarantees and may lead to congested networks in some circumstances.  
The authors of \cite{vne} assume network support for path-splitting \cite{rethinking} in order to use a multi-commodity flow based approach for mapping VMs and flows between them to the physical network, but this approach is not scalable beyond networks containing hundreds of servers \cite{SecondNet}.  

Guo \etal \cite{SecondNet} proposed a new architectural framework, Secondnet, for embedding virtualization requests with bandwidth guarantees. This framework considers requests with bandwidth guarantees $f_{ij}$ between every pair of VMs $(v_i, v_j)$. This framework provides rigorous application performance guarantees and hence is suitable for enterprise workloads but at the same time also establishes hardness of the problem of finding such embeddings in arbitrary networks. 
Our results employ the SecondNet framework but restrict attention to tree networks.

Very recently, Ballani \etal \cite{sig11} have described a {\em virtual cluster} request model, which
consists of requests of the form $\langle k,B\rangle$ representing $k$ VMs each connected to a virtual switch with a link of bandwidth $B$. A request $\langle k,B\rangle$ can be interpreted (although not exactly) as a clique request on $k$ VMs with a bandwidth guarantee of $B/(k-1)$ on each edge of the clique. 
They describe a novel VM allocation algorithm for assigning such requests on a tree network with the goal of maximizing the ability to accommodate future requests. 
For each $v$ in the tree network $T$, they maintain an interval of values that represents the number of VMs that can be allocated to $T_v$ without congesting the uplink edge from $v$ and allocate VMs to sub-trees greedily. We generalize this approach to the case of virtualization requests in the Secondnet framework and we use a dynamic programming solution in order to find the optimal minimum congestion embedding.  
By restricting the requests to {\em virtual clusters}, \cite{sig11} offers a tradeoff to the providers between meeting specific tenant demands and flexibility of allocation schemes. In this work, we explore this tradeoff further and show that it is possible to formulate flexible allocation schemes even in the Secondnet framework for small requests. 
   
The problem of resource allocation has also been studied in the virtual private network (VPN) design setting where bandwidth guarantees are desired between nodes of the virtual network \cite{duffield99, gupta01}. In this setting, a set of nodes of the physical network representing the VPN endpoints is provided as the input, and the task is to reserve bandwidth on the edges of the network in order to satisfy pairwise bandwidth requirements between VPN endpoints. The fixed location of VPN endpoints makes this problem significantly different from that of embedding virtualization requests in a network, since the latter involves searching over all possible embeddings of the VMs in the network. 

\section{Preliminaries}
\label{sec:prelim}
An instance of our problem consists of a datacenter network and a request network.
The datacenter network $N$ is a tree on $n$ nodes rooted at a gateway node $g$.  Edges in $N$ have capacities $c_e$ representing their bandwidth.  Let $L$ denote the set of leaves of $N$.

The request network $G_R$ is an arbitrary, undirected graph on $k+1$ nodes.  Nodes in $G_R$ consist of a set $V$ of $k$ virtual machines $v_1,\ldots, v_k$ and a special gateway node $g$.  Edges $e$ in the request graph specify bandwidth guarantees $f_e$ (flow requirements) and are divided into two types: edges of type-\rn{1} have the form $e=(v_i,g)$ and specify a requirement for routing $f_e$ flow between $v_i$ and the gateway node $g$ (uplink bandwidth to the outside world), and edges of type-\rn{2} have the form $e=(v_i,v_j)$ and specify flows between two virtual machines $v_i$ and $v_j$ (``chatter'' bandwidth between virtual machines).  We use $R^{\rn{1}}$ and $R^{\rn{2}}$ to denote the sets of type-\rn{1} and type-\rn{2} edges and $R = R^{\rn{1}} \cup R^{\rn{2}}$ to denote all edges.

A solution consists of an embedding $\pi: V \rightarrow L$ mapping virtual machines onto leaves in the datacenter network.  For simplicity we will assume only a single VM can be mapped to each leaf, although it is easy to modify our algorithm so that each datacenter node $v$ can support up to $n_v$ VMs.  The gateway node $g$ in $G_R$ is always mapped to the gateway in $N$.  If $\pi$ maps the endpoints of edge $e=(v_i,v_j)$ (equivalently $e=(v_i,g)$) onto $\pi(v_i)$ and $\pi(v_j)$, then $e$ contributes $f_e$ flow to every edge along the path $P_{\pi(v_i),\pi(v_j)}$ between $\pi(v_i)$ and $\pi(v_j)$ in $N$.  The congestion of an edge $e$ in $N$ under embedding $\pi$ is 
\[
\cong(\pi,e) = \frac{1}{c_e} \sum_{(u,v) \in R \textnormal{ s.t. } e \in P_{\pi(u),\pi(v)}} f_{(u,v)}
\]
and our goal is to find $\pi$ minimizing $\max_{e \in N} \cong(\pi,e)$.

%

\section{Hardness results}
\label{sec:hardness}
\newcommand{\e}{\epsilon}

In this section we show that the embedding problem is $\mathcal{NP}$-hard even with the restricted topologies of the host and request graphs.   In particular, we show that the problem of embedding a {\em weighted path request}, which is perhaps the simplest weighted request topology,  is $\mathcal{NP}$-hard to approximate to a factor better than $O(\theta)$, where $\theta$ is the ratio of the largest to smallest bandwidth requirements.  Furthermore, we show that in the unweighted case the problem is $\mathcal{NP}$-hard to approximate to a factor smaller than $O(n^{1-\e})$ for any constant $\e\in (0, 1)$, where $n$ is the number of VMs in the request, even for the case when the request topology is a tree. 

Both of our reductions are from 3-partition.  An instance of 3-partition consists of a multiset $S = \{s_1, \ldots, s_{3m}\}$ of $3m$ integers summing to $mB$, and the goal is to determine whether $S$ can be partitioned into $m$ subsets $S_1, \ldots, S_m$ such that the sums of elements in each of the $S_i$ are equal to $B$ and $|S_i|=3$ for all $i$.  Crucially, 3-partition remains $\mathcal{NP}$-complete even when the size of the integers are bounded by a polynomial in $m$:

\begin{theorem}[\cite{garey-johnson}] The 3-partition problem is strongly $\mathcal{NP}$-complete, even when $B/4<s_i<B/2$ for all $i$, forcing any partition to consist of triples.
\end{theorem}
\renewcommand{\P}{\mathcal{P}}

\subsection{Weighted topologies} 

\begin{theorem} \label{thm:hardness-1} The embedding problem  is $\mathcal{NP}$-complete even when restricted to instances where the request graph is a weighted path, and the host network is a tree. Moreover, it is $\mathcal{NP}$-hard to approximate to within a factor better than $\theta/6$, where $\theta$ is the ratio of the largest to smallest weight in the request graph.
\end{theorem}

\begin{proof}
First, the problem is in $\mathcal{NP}$, since given a candidate embedding, it is easy to verify that its congestion is at most $1$.

Now, let $S = \{s_1, \ldots, s_{3m}\}$ be a multiset of $3m$ integers summing to $mB$, constituting an instance of 3-partition, such that $B/4<s_i<B/2$ for all $i$.  Let $T$ be a tree of height two.  The root/gateway $g$ has $m$ children labeled $S_1, \ldots, S_m$, each of which has $B$ children of its own.  Since 3-partition is strongly $\mathcal{NP}$-complete, we may assume that $B$ is bounded by a polynomial in $m$, so $T$ has polynomial size.  All edges from $g$ to the $S_i$ have capacity 6.  Each node $S_i$ is connected to each of its $B$ children by edges of capacity $W>6$.  

Let $R=R^I\cup R^{II}$ be defined as follows. Let $V=\{v_1,\ldots, v_{mB}\}$ be a set of VMs. For $j=1,\ldots, 3m+1$ let $q_{j}=\sum_{i=0}^{j-1} s_i$, where we set $s_0=0$ for convenience (note that $q_1=0$ and $q_{3m+1}=mB$).  Further, define {\em heavy intervals} as $I_j=\{v_{q_j+1},\ldots, v_{q_{j+1}}\},j=1,\ldots, 3m$, so that $|I_{j}|=s_j$.

 Define chatter bandwidth requests $f_{ij}$ by setting 
\begin{equation*}
f_{ij}=\left\lbrace
\begin{array}{cc}
W,&\text{if $\{i, j\}\subseteq I_k$ for some $k$}\\
1&\text{otherwise.}
\end{array}
\right.
\end{equation*}

Define uplink bandwidths as $f_i=1$ for $i=1$ and $f_i=0$ otherwise. Thus, the requests form a path with the first node on the path connected to the gateway node. The path is partitioned into intervals of length $s_i$, such that the bandwidth requirement between consecutive nodes in each interval is high and the requirement between adjacent nodes on the path that belong to different intervals is low. We refer to the edges of weight $W$ as {\em heavy edges} and the edges of weight 1 as {\em light edges}.

If $S$ has a 3-partition, then the heavy intervals $I_j$ can be divided into $m$ sets $\P_1, \ldots, \P_m$ of 3 intervals each, such that the sum of the lengths within each $\P_i$ is exactly $B$.  We can map all  VMs in $\P_i$ to the children of node $S_i$.  Each edge $(g,S_i)$ carries flow from at most 2 light edges on the border of each of the 3 heavy intervals in $\P_i$, and each edge connecting $S_i$ to its children has load at most $W$, for a congestion of 1.

Now suppose that $S$ does not have a 3-partition. Then since by assumption $B/4<s_i<B/2$, in any feasible allocation of VMs at least one heavy interval $I_k$ must be divided between children of different nodes $S_i$ and $S_j$, and hence at least one heavy edge must congest the edge $(r, S_i)$, yielding congestion at least $W/6$. 

Thus, it is $\mathcal{NP}$-hard to distinguish between instances with an optimal congestion of $1$ and $W/6=\theta/6$, where $\theta$ is the ratio of largest and the smallest weight in the request graph, i.e.\ $\theta=W/1$. 
\end{proof}

\subsection{Unweighted topologies}

\begin{theorem} \label{thm:hardness-2} Let $n$ denote the number of leaves in the host tree. The embedding problem is $\mathcal{NP}$-complete and $\mathcal{NP}$-hard to approximate to within a factor better than $\Omega(n^{1-\e})$, for any $\e\in (0, 1)$, when the set of requests forms an unweighted tree.
\end{theorem}

\begin{proof}
As before, we first note that the problem is in $\mathcal{NP}$, since given a candidate embedding, it is easy to verify  that its congestion is at most $1$.  We use a reduction to 3-partition similar to the reduction to Maximum Quadratic Assignment used in \cite{qa5}.

 Let $S = \{s_1, \ldots, s_{3m}\}$ be a multiset of $3m$ integers summing to $mB$, constituting an instance of 3-partition, such that $B/4<s_i<B/2$ for all $i$.  Let $T$ be a tree of height two.  The root $g$ has $m$ children labeled $S_1, \ldots, S_m$, each of which has $3+B\cdot M$ children of its own, where $M=(5 m B)^{\lceil (1-\e)/\e\rceil}$.  Since 3-partition is strongly $\mathcal{NP}$-complete, we may assume that $B$ is bounded by a polynomial in $m$, so $T$ has polynomial size.  Each node $S_i$ is connected to each of its $3+B\cdot M$ children by links of capacity $B\cdot M+2$, and the root is connected to each of $S_i$ by links of capacity $6$. 

We now define $R=R^I\cup R^{II}$. Let $V=V^1\cup V^2$, where $V^1=\{v_1^1,\ldots, v^1_{m}\}$  and $V^2=\{v^2_1,\ldots, v^2_{mBM}\}$ be a set of VMs organized in a tree as follows. First for $j=1,\ldots, 3m+1$ let $q_{j}=\sum_{i=0}^{j-1} s_i$, where we set $s_0=0$ for convenience. 
We now define bandwidth requirements between VMs in $V$. Each $v^1_i\in V^1$ requires chatter connections of bandwidth 1 to $v^2_{M\cdot (q_i+1)}, v^2_{M\cdot (q_i+1)+1},\ldots, v^2_{M\cdot q_{i+1}}$. Also, $v^1_i$ requires a chatter connection to $v^1_{i-1}$ if $i>1$ and $v^1_{i+1}$ if $i<m$. Finally, both $v^1_1$ and $v^1_m$ require uplink connections to gateway $g$ of bandwidth 1. Thus, the request topology is a tree consisting of stars on $s_i\cdot M$ nodes with centers $v^1_{i}$, for each $i=1,\ldots, m$. Adjacent centers of stars (i.e.\ $v^1_i$ and $v^1_j$ for $|i-j|=1$) are connected to each other.

If $S$ admits a 3-partition, then there exists an embedding of congestion at most $1$: assign the corresponding three centers and their children to the children of $S_j$ for $j=1,\ldots, m$, which is possible since each $S_j$ has exactly $3+B\cdot M$ children.  The congestion is at most $1$ since the edges of $T$ incident on the nodes where the centers are mapped will carry load exactly $B\cdot M+2$ ($B\cdot M$ unit bandwidth connections to the children as well as two connections to neighboring centers or uplink connections), and the edges $(S_j, g)$ will carry at most $2$ units from each of the 3 centers mapped to the children of $S_j$, yielding congestion at most $1$.

Now suppose that $S$ does not admit a 3-partition. Consider the node $S_j\in T$ with the maximum number of centers mapped to its children. Denote these centers by $v^1_{c_1},\ldots, v^1_{c_k}$, where $k>3$.  We then have $\sum_{j=1}^k s_{c_j}\geq B+1$, and hence at least $M$ children of $v^1_{c_1},\ldots, v^1_{c_k}$ are mapped outside the set of children of $S_j$. Hence, at least $M$ edges from the centers $v^1_{c_1},\ldots, v^1_{c_k}$ to these children congest the edge $(S_i, g)$, where $g$ is the root of $T$. Thus, the congestion is at least $M/6$. The number of vertices in the tree $T$ is $n=1+m(3+B\cdot M)\leq 1+(3+B)m\cdot M\leq (5mB)\cdot M\leq M^{\e/(1-\e)+1}=M^{1/(1-\e)}$. Hence, the congestion is at least $M/6\geq n^{1-\e}/6$.

We have shown that it is $\mathcal{NP}$-hard to distinguish between instances of the problem where the minimum congestion is $1$ and $\Omega(n^{1-\e})$, thus completing the proof.
\end{proof}

\section{Algorithm}\label{sec:algo} 

Next we present our algorithmic results and show that despite the $\mathcal{NP}$-completeness results in the previous section, many practical instances can still be solved efficiently.

\subsection{Creation of binary tree}

We first convert the tree $N$ into a binary tree $T$ with not many additional nodes in a way that preserves the congestion of all solutions.  This step is purely for convenience in simplifying the presentation of the algorithm that follows.
We simply replace each degree $d$ node with a complete binary tree on $d$ nodes.  Algorithm \ref{binary_tree_alg} describes the procedure Create-Binary-Tree($N,g$) more formally.

\begin{algorithm}
\caption{Create-Binary-Tree($N,g$)}
\begin{algorithmic}[1]
 \label{binary_tree_alg}
\FORALL { $v \in N$, $\degree(v) > 3$} 
	\STATE Let $u_1,\ldots, u_d$ be the children of $v$, and $e_1,\ldots, e_d$ the edges connecting $v$ to $u_i$ 
	\STATE Replace $e_1,\ldots,e_d$ with a binary tree rooted at $v$ with leaves $u_1,\ldots, u_d$ 
	\STATE Set the capacity of parent edges of $u_i$ to be $c_{e_i}$ and that of all other new edges to be $\infty$ 
\ENDFOR
\end{algorithmic}
\end{algorithm}

Let $T$ be the resulting binary tree.  We first show that the congestion of embedding into $T$ and $N$ is equal:

\begin{lemma}\label{lem:equivalence}
The congestion of embedding any request graph $G_R$ into a tree $N$ rooted at node $g$ is equal to the congestion of embedding $G_R$ into the binary tree $T$ constructed by the procedure Create-Binary-Tree($N,g$)
\end{lemma}
\begin{proof}
Consider any embedding $\pi$ of $G_R$ into $N$. 
Since the auxiliary nodes inserted are not leaves, $\pi$ defines an embedding of $G_R$ into $T$ as well.  Let $u,v \in N \cap T$ and $P_{u,v}^N$, $P_{u,v}^T$ be the edges on the unique paths between $u$ and $v$ in $N$ and $T$.  Observe that $P_{u,v}^N \subseteq P^{T}_{u,v}$, and that all edges in $P^{T}_{u,v} \setminus P_{u,v}^N$ have infinite capacity and contribute nothing to the congestion.
Hence the congestion of embedding in $N$ and $T$ is equal. 
\end{proof}

Next, we show that $T$ is not much bigger than $N$:

\begin{lemma} 
\label{lem:size}
The number of nodes is $T$ is at most $2n$ and the height of $T$ is $O(H\log \Delta)$ where $\Delta$ is maximum degree in $N$ and $H$ denotes the height of $N$.
\end{lemma}
\begin{proof}
We replace each node $v$ of degree $d$, with a complete binary tree on $d$ leaves, which has at most $2d$ nodes. Therefore, the number of nodes in $T$ is at most $2n$. Also by this replacement, we stretch sub-trees of height $1$ by a factor at most $\lceil \log \Delta \rceil$ which shows that the height of $T$ is $O(H\log \Delta)$.
\end{proof}

\subsection {Minimum congestion of embedding requests in a binary tree}

\renewcommand{\part}{\textrm{Part}}
\newcommand{\flow}{\textrm{Flow}}

Now we present our primary algorithmic result and show that if the request graph is small---which is true in many practical instances---then the optimal embedding can be found efficiently.
Before describing the algorithm, we introduce some notation. For any node $u \in T$ we use the symbol $e_u$ to denote the link joining the parent of node $u$ to $u$ and $T_u$ to denote the subtree of $T$ rooted at $u$.  If $u$ is not a leaf, we refer to the ``left'' and ``right'' children of $u$ in $T$ as $u_l$ and $u_r$ respectively.  In this section we assume that the tree $T$ rooted at $g$ is binary and of height $H$. Let $L^j$ denote the set of vertices in $T$ at distance $j$ from $g$, so $L^0 = \{g\}$, while $L^H$ denotes the leaves at the lowest level.  


The algorithm is straightforward dynamic programming.  Starting at the leaves of $T$, and moving upwards towards the root, for each node $u \in T$ and set $S \subseteq V$ we calculate the congestion of the optimal embedding of the VMs in $S$ into $T_u$ using the congestion of embeddings into $u$'s children.  Let $\flow[S]$ denote sum of the bandwidth requirements crossing the cut $(S,V\cup\{g\}\setminus S)$ in $G_R$, and
$\cong[u,S]$ denote the optimal congestion of the edges of $T_u$ when embedding the subgraph of $G_R$ spanned by $S$ into $T_u$.
Then $\cong[u,S]$ satisfies the recurrence
\[
\cong[u,S] = \min_{S_l \subseteq S} \max\left\{ 
\cong[u_l,S_l],\cong[u_r,S\setminus S_l], 
 \flow[S_l]/c_{e_l},\flow[S\setminus S_l]/c_{e_r} \right\}
\]
That is, it is the minimum over all partitions $(S_l,S\setminus S_l)$ of $S$ of the congestion of embedding $S_l$ into $T_{u_l}$ and $S\setminus S_l$ into $T_{u_r}$.  The terms $\flow[S_l]/c_{e_l}$ and $\flow[S\setminus S_l]/c_{e_r}$  are the congestion on the edges connecting $u$ to its children.  The base case is when $u$ is a leaf, in which case
\[
\cong[u,S] = \begin{cases}
0 & \textnormal{if } |S| \le 1 \\
\infty & \textnormal{if } |S| > 1
\end{cases}
\]
assuming for simplicity that each server can support at most a single VM.  By changing this equation, we can easily allow a server $v \in T$ to support up to $n_v$ VMs.

After computing these recurrences, the  algorithm outputs $\cong[g,V]$. Note that $L^0 = \{g\}$ and that it suffices to compute $\cong[g,V]$ (i.e., $\cong[g,S]$ for subsets $S \subset V$ is not needed).  Algorithm \ref{alg:cong} shows the procedure in more detail.

\renewcommand{\algorithmicrequire}{\textbf{Input:}}
\renewcommand{\algorithmicensure}{\textbf{Output:}}

\begin{algorithm}
\caption{Minimum Congestion}
\label{alg:cong}
\begin{algorithmic}[1]
\REQUIRE Binary tree $T$ rooted at $g$, request graph $G_R$ 
\ENSURE Minimum congestion in embedding $V$ into $T$ such that requirements $R$ are satisfied 
\FORALL {$S \subseteq V$}
	\STATE  $\begin{aligned} 
	\flow[S] \gets \sum_{(v,g) \in R^{\rn{1}}, v \in S} f_{(v,g)} + \sum_{(u,v) \in R^{\rn{2}}, u\in S, v\notin S} f_{(u,v)}
	\end{aligned}$
\ENDFOR

\FORALL {leaves $u \in L$, and $S \subseteq V$ }
\STATE $\cong[u,S] \gets 0 \text{ if } |S| \leq1, \infty \text{ otherwise}$ 
\ENDFOR
\FOR {$j = H,H-1,\ldots, 0$}
	\FORALL {$ u \in L^j$,  $u$ not a leaf}
		\STATE $t_{\min} \gets \infty$
		\FORALL {$S \subseteq V$} 
		\FORALL { $S_l \subseteq S$}
		\STATE
				$t \gets \max 
				\left\{
				 \cong[u_l, S_l], 
				\cong[u_r, S \setminus S_l], 
				\flow[S_l]/c_{e_l},  
				\flow[S\setminus S_l]/c_{e_r} 
				\right\}$

		\IF {$t < t_{\min}$}
			\STATE $t_{\min} \gets t$
			\STATE $S_{\min} \gets S_l$
		\ENDIF				
		\ENDFOR
		\STATE $\cong[v,S] \gets t_{\min}$
		\STATE $\part[u,S] \gets (S_{\min}, S \backslash S_{\min})$
		\ENDFOR
	\ENDFOR
\ENDFOR
\RETURN $\cong[g,V]$
\end{algorithmic}
\end{algorithm}

When we update $\cong[u,S]$ we also store the partition $(S_l,S\setminus S_l)$ that realizes this optimal congestion in a partition table $\part[u,S]$.  After the execution of the algorithm, we can recover the optimal embedding by working backwards in the standard fashion for dynamic programs: starting at $g$ we read the optimal partition $(V_1, V\setminus V_1)$ from $\part[g,V]$.  Now we find the optimal partitions of $V_1$ with root $g_l$ and $V\setminus V_1$ with root $g_r$, and so on.


Now we analyze the correctness and runtime:

\begin{lemma}\label{lem:correctness}
Algorithm \ref{alg:cong} finds the minimum congestion of embedding request $G_R$ in a tree network $N$.
\end{lemma}
\begin{proof}
By Lemma \ref{lem:equivalence}, optimizing the congestion on $N$ is equivalent to optimizing it on $T$.
The optimal congestion of an embedding restricted to $T_u$ requires using an optimal partition into subsets embedded into left and right subtrees of $T_u$, and Algorithm \ref{alg:cong} recursively computes the optimal embedding for all possible partitions of the VMs, thus retrieving the congestion of the optimal embedding. 
\end{proof}

\begin{lemma}
Algorithm \ref{alg:cong} has running time $O(3^kn)$. 
\end{lemma}
\begin{proof}
We first calculate $\flow[S]$ for every set $S \subseteq V$.  There are $2^k$ such sets, and each requires summing over at most $k^2$ edges in $R$, for a runtime of $O(k^22^k)$, which is $O(3^k)$ for large enough $k$.
In the main loop, for each $u$ in $T$ we compute $\cong[u,S]$ for all sets $S \subseteq V$.  If $|S| = i$, computing $\cong[u,S]$ requires looking at all $2^i$ subsets of $S$ and doing $O(1)$ work for each one.  Summing over all $O(n)$ nodes and all sets $S$, this requires $O(n)O(\sum_{i=0}^k {k \choose i} 2^i) = O(3^kn)$ work total.
\end{proof}

\subsection{Other Objective Functions and Request Models} \label{subsec:virtual_cluster}

The basic form of our algorithm is not specific to congestion, and the recurrence in Algorithm \ref{alg:cong} can easily be modified to optimize for any objective function for which we can write a similar recurrence.  For instance, if each edge in $T$ has a delay and bandwidth capacity, we can minimize the average or maximum latency between VMs subject to satisfying bandwidth constraints (with a slightly more complex recurrence).

In practice it may not be desirable to allow request graphs to have arbitrary topologies and edge weights.  If a request graph is sufficiently simple and uniform, then the complexity results of Section \ref{sec:hardness} no longer apply, and we no longer need to consider all $2^k$ cuts of $G_R$ at each node.  For instance, if $G_R$ is a clique with equal bandwidth on all edges, then the congestion of embedding a set of VMs $S$ into $T_u$ is dependent only on the size of $S$, so we only need to consider $k+1$ subproblems for each node in $T$.

Ballani et al.\ \cite{sig11} describe a {\em virtual cluster} request model, which
consists of requests of the form $\langle k,B\rangle$ representing $k$ VMs each connected to a virtual switch with a link of bandwidth $B$. Such a request $\langle k,B\rangle$ is similar (but not identical) to a request consisting of a clique on $k$ VMs and a bandwidth guarantee of $B/(k-1)$ on each edge of the clique in our setting. 
We show that when restricted to {\em virtual cluster} requests, a modified version of Algorithm \ref{alg:cong} finds the minimum congestion embedding in time $O(nk^2)$. For the sake of completeness and comparison with their work, we present Algorithm \ref{alg:virtual-cluster}. Similar adjustments could be made to handle other request models for which considering all $2^k$ cuts of the request graph is unnecessary.

\renewcommand{\algorithmicrequire}{\textbf{Input:}}
\renewcommand{\algorithmicensure}{\textbf{Output:}}

\begin{algorithm}
\caption{Min Congestion Embedding for $\langle k,B\rangle$}
\label{alg:virtual-cluster}
\begin{algorithmic}[1]
\FORALL {leaves $u \in L$, and $i \in 0,\ldots,k$ }
\STATE $\cong[u,i] \gets 0 \text{ if } z \le 1, \infty \text{ otherwise}$ 
\ENDFOR
\FOR {$j = H,H-1,\ldots, 1$}
        \FORALL {$ u \in L^j$}
                \STATE $t_{\min} \gets \infty$
                \FOR {$z = 0, \ldots, k$} 
                \FOR { $i = 0, \ldots, z$}
                \STATE $f_l \gets i\cdot(k-i)\cdot B/(k-1)$
                \STATE $ f_r \gets (z-i) \cdot (k-z+i) \cdot B/(k-1)$                
                \STATE 
                                $  t \gets \max
  				\left\{
                                \cong[u_1, i], 
                                \cong[u_2, z-i], 
   		  		f_l/c_{e_l}, 
                                f_r/c_{e_r}
                                \right\}$

		\IF {$t < t_{\min}$}
                        \STATE $t_{\min} \gets t$
   			\STATE $i_{\min} \gets i$
                \ENDIF 
                \ENDFOR
                \STATE $\cong[u,z] \gets t_{\min}$
 		\STATE $\part[u,z] \gets (i_{\min}, z-i_{\min})$
                \ENDFOR
        \ENDFOR
\ENDFOR
\RETURN $\cong[g,k]$
\end{algorithmic}
\end{algorithm}

The correctness of Algorithm \ref{alg:virtual-cluster} can be inferred from the correctness of Algorithm \ref{alg:cong} by noting that under the {\em virtual cluster} request model, all subsets of equal size embed in a subtree with same congestion, i.e.\ for any $S_1, S_2 \subseteq V$ such that $|S_1| = |S_2|$, we have $\cong[u,S_1] = \cong[u, S_2]$ for all $u \in T$. For every node $u$ and for all $z \in 0 \ldots k$, Algorithm \ref{alg:virtual-cluster} calculates $\cong[u,z]$ by optimizing over $z+1$ possible splits of the $z$ VMs among its children. A simple recursive calculation shows that this computation has complexity $\sum_{z = 0}^k (z+1) = O(k^2) $. This shows that the running time of Algorithm \ref{alg:virtual-cluster} is $O(nk^2)$.  



\section{Simulations}\label{sec:simu}
In this section we present simulations which verify the correctness and scaling properties of Algorithm \ref{alg:cong} in both the pairwise bandwidth guarantees model, as well as {\em virtual cluster} request model \ref{subsec:virtual_cluster}.  
We perform all simulations using an unoptimized python implementation of Algorithm \ref{alg:cong} on Intel Sandy Bridge Quad Core machine having $4$ GB of RAM using the networkx graph package \cite{networkx} to simulate the physical network. 
\begin{figure*}[!t]
\centerline{\subfloat[Linear variation with $n$] {
	\includegraphics[width=2.5in]{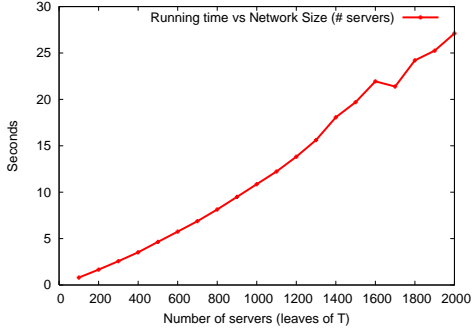}
	}
\hfil
\subfloat[Exponential variation with $k$] {
	\includegraphics[width=2.5in]{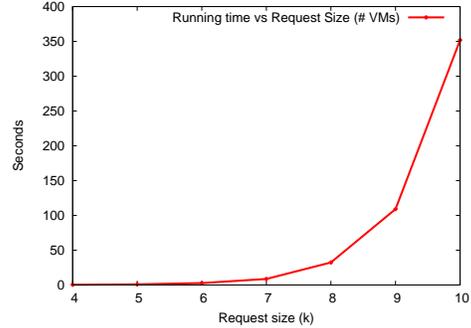}
	}
}
\caption{Pairwise bandwidth guarantees between all VMs: Dependence of the running time of Algorithm \ref{alg:cong} with (a) $n$, size of the network when $k = 5$, (b) and with $k$, size of the requests when $n = 100$. }
\label{fig:scale}
\end{figure*}

\subsection{Network configuration} 

In order to test our algorithm on a realistic networks, we simulate a typical three tier data center network \cite{alfares} with servers housed in racks which are connected to a Top-Of-Rack (TOR) switch (tier I). The TOR switches connect the racks to other parts of the network via Aggregation Switches (AS, tier II).  The AS switches have uplinks connecting them to the Core Switch (CS, tier III).  We assume that TOR's are connected to the servers with $10$ GBps links while the uplinks from TORs to the AS's are $40$ GBps and from the AS's to the CS's are $100$ GBps. We construct a tree topology over these elements, recalling that common routing protocols used in practice employ a spanning tree of the physical network. We model existing traffic in the data center network using random residual capacities for each link. We choose the residual capacity for edge $e$ independently of all other edges and uniformly at random from $[0,c(e)]$ where $c(e)$ denotes the bandwidth capacity of edge $e$. The choice of random residual link capacities is forced on us due to lack of models describing realistic network flows in a data center. We note that Algorithm \ref{alg:cong} finds the optimal congestion embedding for any of the distribution of residual capacities on the network links and any choices of bandwidth capacities of the links.  

\subsection{Linear scan over all possible VM allocations}
By implementing a linear scan over all possible VM allocations in the network, we verify the correctness of Algorithm \ref{alg:cong} by finding the allocation that minimizes congestion. Note that this implementation requires scanning ${n \choose k} \cdot k!  = O(n^k)$ feasible VM allocations where $n$ denotes the number of servers in the network and $k$ denotes the request size. Hence we choose small network and request sizes $n \in \{50,75,100\}$ and $k = 4$ and verify correctness of the algorithm for different request topologies and randomly generated residual capacities on the network links. We observe that this procedure requires hours or even days to finish even for very small network and request sizes like $n = 125$ and $k = 4$ as seen in Table \ref{table:brute-force} and hence is infeasible for modern data centers containing hundreds of thousands of servers. In contrast, Algorithm \ref{alg:cong} has complexity $O(3^kn)$, which is linear in the network size $n$, and as shown in the next sub section, finishes in order of seconds on our simulation setup for small values of $k$. 

\begin{table}[!t]
\caption{Linear scan for VM allocation: run time} 
\label{table:brute-force}
\centering
\begin{tabular}{c|c|c}
n & k & Time (hours) \\ \hline
50 & 4 & 2.2  \\
75 & 4 & 18.8 \\ 
100 & 4 & 80 
\end{tabular}
\end{table}

\subsection{Pairwise bandwidth requirements} 
Next, we verify the scaling properties of Algorithm \ref{alg:cong} with respect to parameters $n$ and $k$. First, we fix a request of size $k = 5$, and plot the running time for increasing values of $n$, the number of servers, from $n = 200$ to $n = 2000$ in Figure \ref{fig:scale}(a) which illustrates the linear variation of run time with respect to $n$. Next, we fix the network size to $n = 100$ and plot the run time for {\em path requests} with lengths from $k = 4$ to $k = 10$ in Figure \ref{fig:scale}. This figure shows that the run time increases exponential with respect to $k$.  
 
\subsection{Virtual Cluster Request Model}
We also verify the scaling properties when the requests are restricted to the virtual cluster model \cite{sig11}. 
For a fixed request $\langle k,B\rangle$ where $k = 100$ and $B = 100$Mbps, we plot the running time for increasing values of $n$, from $n = 200$ to $n = 2000$ in Figure \ref{fig:scale:vcluster}(a) which illustrates the linear variation of run time with $n$. 
Next, we fix the network size to $n = 1000$ and plot the run time for $k$ in the range $10$ to $100$. These results show that for {\em virtual cluster} requests, our algorithm finds the minimum congestion embedding in time $O(nk^2)$.  

\begin{figure*}[!t]
\centerline{\subfloat[Linear variation with $n$] {
	\includegraphics[width=2.5in]{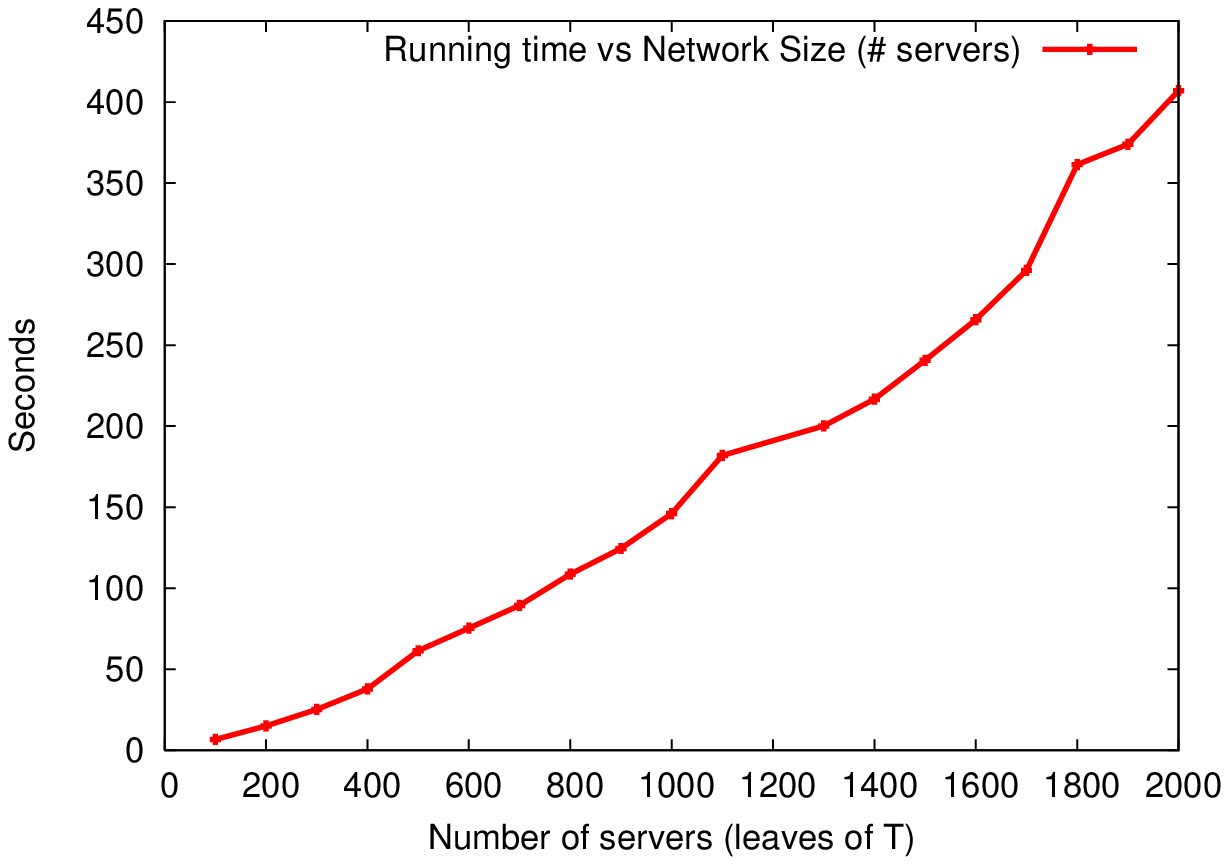}
	}
\hfil
\subfloat[Quadratic variation with $k$] {
	\includegraphics[width=2.5in]{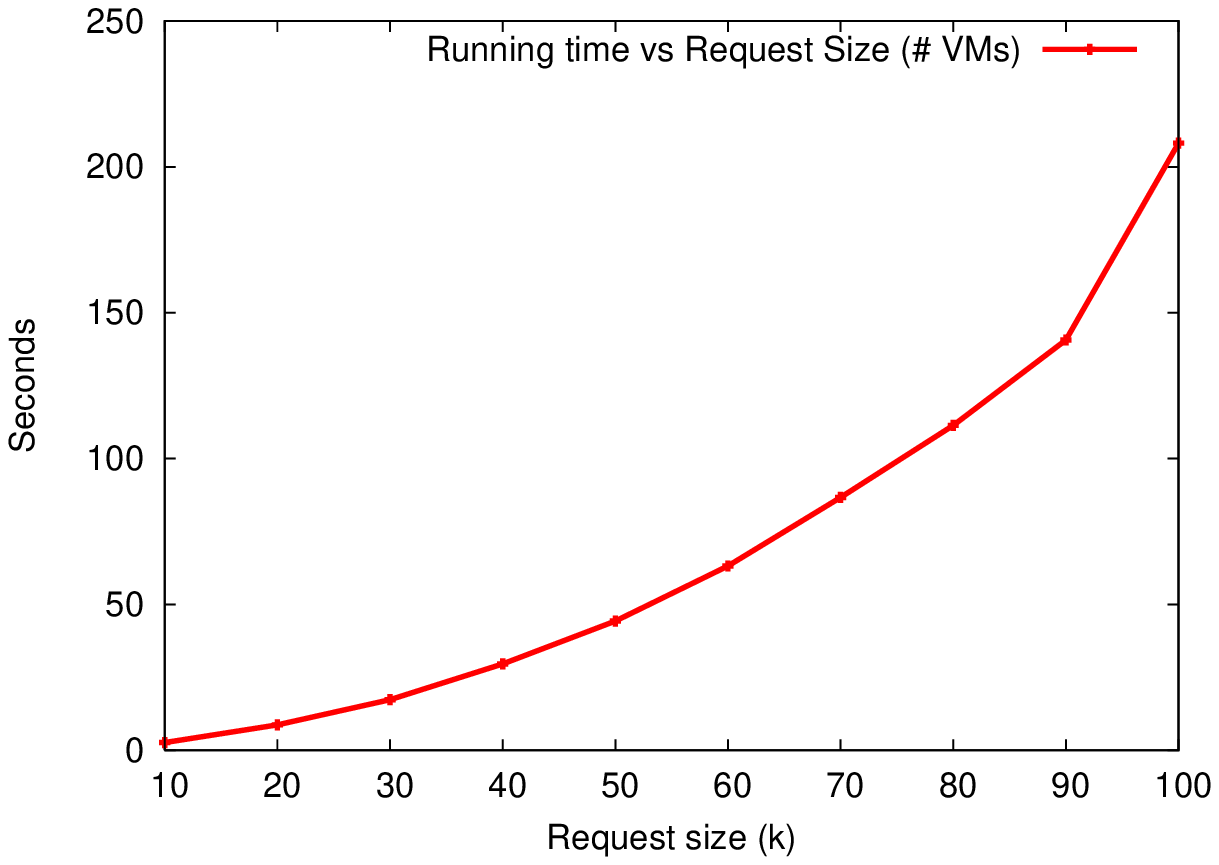}
	}
}
\caption{{\em Virtual Cluster} request model: Dependence of the running time of Algorithm \ref{alg:virtual-cluster} for the {\em virtual cluster} request model with (a) $n$, size of the network when $k = 10$, (b) and with $k$, size of the requests when $n = 1000$. }
\label{fig:scale:vcluster}
\end{figure*}

As mentioned before, a number of heuristic approaches have been formulated to perform VM allocation. However, lack of models of the existing network flows inside a data center, especially in the context of enterprise workloads, hinders the evaluation and comparison of their performance in realistic settings. In particular, we observe that by congesting particular edges in the network, it is possible to make the greedy heuristics for VM mapping perform significantly worse than the optimal embedding (output by Algorithm \ref{alg:cong}). However, a thorough comparison with heuristics requires models of flow in a data center serving enterprise requests, and we leave this to future work. 

\section{Conlusion and Future Work}
\label{sec:conclusion}

In this paper we study the problem of allocating a graph request within a tree topology and we present a $O(3^kn)$ dynamic programming algorithm that embeds the resource request graph of size $k$ into the data center topology (tree) of size $n$ to minimize congestion. We believe this is useful in enterprise workloads when the request size $k$ is small. For clique requests, we present a $O(n^2k)$ dynamic programming algorithm to allocate clusters of size $k$ in a tree of size $n$ for minimizing congestion, which could be useful for MapReduce-like workloads. We believe that it would also be possible to extend our results to hybrid workloads involving tiers of VMs, with both inter-tier as well as intra-tier bandwidth guarantees. We also provide hardness results and show that the problem of finding minimum congestion embedding in a network remains in $\mathcal{NP}$-hard even under the restriction to tree network. We focus on minimizing congestion as our objective function, but we believe our methods are applicable to a wider class of metrics and objective functions.

\bibliographystyle{alpha}
\bibliography{DP_algorithm}
\end{document}